\documentclass[3p,10pt]{elsarticle}

\usepackage{amssymb,latexsym,amsthm,url,amsmath,subfigure}

\journal{}

\newcommand{\set}[1]{\left\{#1\right\}}

\newcommand{\p}{\partial}

\newcommand{\mU}{\mathbf{U}}
\newcommand{\mV}{\mathbf{V}}

\newcommand{\mx}{\mathbf{x}}

\newcommand{\mz}{\mathbf{z}}

\newcommand{\vt}{\boldsymbol{\theta}}

\newtheorem{thm}{Theorem}[section]

\newtheorem{lem}[thm]{Lemma}

\begin{document}

\begin{frontmatter}



\title{Analysis of subspace migration in the limited-view inverse scattering problems}

\author{Young Mi Kwon}
\ead{himang0912@kookmin.ac.kr}
\author{Won-Kwang Park\corref{Parkwk}}
\ead{parkwk@kookmin.ac.kr}
\address{Department of Mathematics, Kookmin University, Seoul 136-702, Korea.}
\cortext[Parkwk]{Tel: +82 2 910 5748; fax: +82 2 910 4739.}

\begin{abstract}
In this paper, we analyze the subspace migration that occurs in limited-view inverse scattering problems. Based on the structure of singular vectors associated with the nonzero singular values of the multi-static response matrix, we establish a relationship between the subspace migration imaging function and Bessel functions of integer order of the first kind. The revealed structure and numerical examples demonstrate why subspace migration is applicable for the imaging of small scatterers in limited-view inverse scattering problems.
\end{abstract}

\begin{keyword}
Subspace migration \sep limited-view inverse problem \sep Multi-Static Response (MSR) matrix \sep Bessel functions \sep numerical examples



\end{keyword}

\end{frontmatter}


\section{Introduction}
This work focuses on subspace migration for the imaging of small perfectly conducting cracks located in the homogeneous, two-dimensional space $\mathbb{R}^2$ in limited-view inverse scattering problems. Subspace migration is a simple and effective imaging technique, and is thus applied to many inverse scattering problems. For this reason, various remarkable properties have been investigated in many studies. Related work can be found in \cite{A,AGKPS,HHSZ,JKHP,P1,P3,P4,PL2,PP} and the references therein.

Based on these studies, it turns out that subspace migration can be applied not only in full-view but also limited-view inverse scattering problems. However, the use of subspace migration-based imaging techniques has been applied heuristically. Although the structure of subspace migration in full-view inverse scattering problems has been established~\cite{A,JKHP}, it is still applied heuristically in limited-view problems. Hence, a mathematical identification of its structure in such problems still needs to be performed, which is the motivation for our work.

In this paper, we extend the analysis \cite{JKHP} of the structure of subspace migration in the full-view problem to the limited-view problem of imaging perfectly conducting cracks of small length. This is based on the fact that the far-field pattern can be represented by the asymptotic expansion formula in the presence of such cracks. From the identified structure, we find a condition for successful imaging. Subsequently, the structure of multi-frequency subspace migration is analyzed. This shows how the application of multiple frequencies enhances its applicability.

This paper is organized as follows. In Section \ref{sec2}, we survey two-dimensional direct scattering problems, the asymptotic expansion formula in the presence of small cracks, and subspace migration. In Section \ref{sec3}, we investigate the structure of single- and multi-frequency subspace migration in limited-view problems by finding a relationship with integer-order Bessel functions of the first kind. Section \ref{sec4} presents some numerical examples to support our investigation.

\section{Preliminaries}\label{sec2}
In this section, we briefly introduce two-dimensional direct scattering problems in the presence of small, linear perfectly conducting crack(s), the asymptotic expansion formula, and subspace migration.

\subsection{Direct scattering problems and the asymptotic expansion formula}
Let $\Sigma$ be a straight line segment that describes a perfectly conducting crack of small length $2\ell$,
\[\Sigma(\mz):=\set{\mz=(x,y)^T:-\ell\leq x\leq\ell}.\]
With this $\Sigma$, let $u_{\mathrm{tot}}(\mx,\vt;k)$ satisfy the following Helmholtz equation
\begin{equation}\label{Helmholtz}
  \left\{\begin{array}{rcl}
            \triangle u_{\mathrm{tot}}(\mx,\vt;k)+k^2u_{\mathrm{tot}}(\mx,\vt;k)=0 & \mbox{in} & \mathbb{R}^2\backslash\Sigma \\
            \noalign{\medskip}u_{\mathrm{tot}}(\mx,\vt;k)=0 & \mbox{on} & \Sigma,
          \end{array}
  \right.
\end{equation}
where $k=2\pi/\lambda$ denotes a strictly positive wavenumber with wavelength $\lambda$ and we assume that $k^2$ is not an eigenvalue of (\ref{Helmholtz}).

It is well known that the total field $u_{\mathrm{tot}}(\mx,\vt;k)$ can be decomposed as
\[u_{\mathrm{tot}}(\mx,\vt;k)=u_{\mathrm{inc}}(\mx,\vt;k)+u_{\mathrm{scat}}(\mx,\vt;k),\]
where $u_{\mathrm{inc}}(\mx,\vt;k)=e^{ik\vt\cdot\mx}$ is the incident field with direction $\vt_n$ on the two-dimensional unit circle $\mathbb{S}^1$, and $u_{\mathrm{scat}}(\mx,\vt;k)$ is the unknown scattered field satisfying the Sommerfeld radiation condition
\[\lim_{|\mx|\to\infty}|\mx|^{1/2}\left(\frac{\p u_{\mathrm{scat}}(\mx,\vt;k)}{\p|\mx|}-iku_{\mathrm{ scat}}(\mx,\vt;k)\right)=0,\]
uniformly in all directions $\hat{\mx}=\mx/|\mx|$.

The far-field pattern $u_\infty(\hat{\mx},\vt;k)$ of the scattered field $u_{\mathrm{scat}}(\mx,\vt;k)$ is defined on $\mathbb{S}^1$ and can be represented as
\[u_{\mathrm{ scat}}(\mx,\vt;k)=\frac{e^{ik|\mx|}}{|\mx|^{1/2}}\left(u_\infty(\hat{\mx},\vt;k)+O\left(\frac{1}{|\mx|}\right)\right)\]
uniformly in all directions $\hat{\mx}=\mx/|\mx|$ and $|\mx|\longrightarrow+\infty$. Based on \cite{AKLP}, the far-field pattern can then be represented as the following asymptotic expansion formula.

\begin{lem}[Asymptotic expansion formula in the presence of a small crack] Let $u_{\mathrm{tot}}(\mx,\vt;k)$ satisfy (\ref{Helmholtz}) and $u_{\mathrm{inc}}(\mx,\vt;k)=e^{ik\vt_n\cdot\mx}$. Then, for $0<\ell<2$, the following asymptotic expansion holds:
  \begin{equation}\label{AsymptoticFormula}
    u_\infty(\hat{\mx},\vt;k)=-\frac{2\pi}{\ln(\ell/2)}u_{\mathrm{inc}}(\mz,\vt;k)\overline{u_{\mathrm{ inc}}(\mz,\hat{\mx};k)}+O\bigg(\frac{1}{|\ln\ell|^2}\bigg).
  \end{equation}
\end{lem}

\subsection{Introduction to subspace migration}
At this point, we apply (\ref{AsymptoticFormula}) to explain an imaging technique known as subspace migration. From \cite{AGKPS}, subspace migration is based on the structure of singular vectors of the collected Multi-Static Response (MSR) matrix
\[\mathbb{K}(k)=\bigg[u_\infty(\hat{\mx}_m,\vt_n;k)\bigg]_{m,n=1}^{N},\]
whose elements $u_\infty(\hat{\mx}_m,\vt_n;k)$ make up the far-field pattern with incident number $n$ and observation number $m$.
From now on, we assume that there exist $S$ different small cracks $\Sigma_s$ with the same length $\ell$, centered at $\mz_s$, $s=1,2,\cdots,S$. The asymptotic expansion formula (\ref{AsymptoticFormula}) can then be represented as follows:
\begin{equation}\label{AsymptoticFormula2}
    u_\infty(\hat{\mx},\vt;k)\approx-\frac{2\pi}{\ln(\ell/2)}\sum_{s=1}^{S}u_{\mathrm{inc}}(\mz_s,\vt_n;k)\overline{u_{\mathrm{ inc}}(\mz_s,\hat{\mx}_m;k)}=-\frac{2\pi}{\ln(\ell/2)}\sum_{s=1}^{S}e^{ik(\vt_n-\hat{\mx}_m)\cdot\mz_s}.
\end{equation}

\subsection{Introduction to subspace migration}
The subspace migration imaging algorithm is based on the structure of the singular vectors of the MSR matrix $\mathbb{K}(k)$. For the sake of simplicity, suppose that the incident and observation directions coincide. In this case, for each $\hat{\mx}_m=-\vt_m$, the $mn-$th element of the MSR matrix becomes
\[u_\infty(-\vt_m,\vt_n;k)=-\frac{2\pi}{\ln(\ell/2)}\sum_{s=1}^{S}e^{ik(\vt_n+\vt_m)\cdot\mz_s}.\]
Now, let us define a unit vector
\[\mathbf{W}(\mx;k):=\frac{1}{N}\bigg[e^{ik\vt_1\cdot\mx},e^{ik\vt_2\cdot\mx},\cdots,e^{ik\vt_N\cdot\mx}\bigg]^T\]
such that $\mathbb{K}(k)$ can be decomposed as
\[\mathbb{K}(k)=-\frac{2\pi}{\ln(\ell/2)}\sum_{s=1}^{S}\mathbf{W}(\mz_s;k)\mathbf{W}(\mz_s;k)^T.\]

Using this decomposition, we can introduce the subspace migration imaging algorithm. By performing Singular Value Decomposition (SVD) on $\mathbb{K}(k)$
\[\mathbb{K}(k)=\mathbb{U}(k)\mathbb{D}(k)\mathbb{V}(k)^*\approx\sum_{s=1}^{S}\sigma_s(k)\mU_s(k)\mV_s(k)^*,\]
we can observe that the left- and right-singular vectors, $\set{\mU_s(k)}_{s=1}^{S}$ and $\set{\mV_s(k)}_{s=1}^{S}$, respectively, satisfy
\[\mU_s(k)\simeq\mathbf{W}(\mz_s;k)\quad\mbox{and}\quad\overline{\mV}_s(k)\simeq\mathbf{W}(\mz_s;k).\]
Hence, by defining an imaging functional
\[\mathcal{I}(\mx;k):=\left|\sum_{s=1}^{S}\left\langle\mathbf{W}(\mx;k),\mU_s(k)\right\rangle\left\langle\mathbf{W}(\mx;k),\overline{\mV}_s(k)\right\rangle\right|\]
we can obtain the shape of $\Sigma_s$ based on the orthogonality of the singular vectors. Here, $\langle\mathbf{a},\mathbf{b}\rangle=\overline{\mathbf{a}}\cdot\mathbf{b}$. A more detailed discussion can be found in \cite{AGKPS,JKHP,P1,P3,P4,PL2,PP}.

\section{Structure and properties of subspace migrations in limited-view problems}\label{sec3}
In this section, we carefully explore the structure of subspace migration in limited view problems and discuss its properties. Suppose that each $\vt_n$ is an element of $\mathbb{S}_{\mathrm{sub}}^1\subset\mathbb{S}^1$ such that
\[\vt_n:=\left(\cos\theta_n,\sin\theta_n\right)^T,\quad\theta_n=\alpha+(\beta-\alpha)\frac{n-1}{N-1},\]
where $\alpha\ne0$ and $\beta\ne2\pi$. From this, we will derive a useful approximation in the following Lemma. This will play a key role in our exploration.

\begin{lem}\label{Lemma}
  Assume that $k$ is sufficiently large, $\vt\in\mathbb{S}_{\mathrm{sub}}^1$, and $\mx,\mz\in\mathbb{R}^2$. Then, the following approximation holds:
\begin{align*}
\int_{\mathbb{S}_{\mathrm{sub}}^1}e^{ik\vt\cdot(\mx-\mz)}dS(\vt)=(\beta-\alpha)J_0(k|\mx-\mz|)+O\bigg(\frac{1}{\sqrt{k|\mx-\mz|}}\bigg).
\end{align*}
\end{lem}
\begin{proof}
Let us consider polar coordinates, such that $\vt=(\cos\theta,\sin\theta)$ and $\mx-\mz=r(\cos\phi,\sin\phi)$. Because the Jacobi--Anger expansion
\[e^{iz\cos\theta}=J_0(z)+2\sum_{n=1}^{\infty}i^{n}J_n(z)\cos(n\theta),\]
holds uniformly (see \cite{DR}), we can evaluate
\begin{align*}
\int_{\mathbb{S}_{\mathrm{sub}}^1}e^{ik\vt\cdot(\mx-\mz)}dS(\vt)&=\int_{\alpha}^{\beta}e^{ikr\cos(\theta-\phi)}dS(\vt)=\int_{\alpha}^{\beta}J_0(kr)d\theta+2\int_{\alpha}^{\beta}\sum_{n=1}^{\infty}i^{n}J_n(kr)\cos(n\theta)d\theta\\
&=(\beta-\alpha)J_0(kr)+2\sum_{n=1}^{\infty}i^{n}J_n(kr)\int_{\alpha}^{\beta}\cos(n\theta)d\theta\\
&=(\beta-\alpha)J_0(kr)+4\sum_{n=1}^{\infty}\frac{i^{n}}{n}J_n(kr)\sin\frac{1}{2}\bigg(n(\beta-\alpha)\bigg)\cos\frac{1}{2}\bigg(n(\beta+\alpha)\bigg).
\end{align*}
Assume that $k|\mx-\mz|$ is sufficiently large. As the following asymptotic form holds
\begin{equation}\label{AsymptoticBessel}
  J_n(k|\mx-\mz|)=\sqrt{\frac{2}{\pi k|\mx-\mz|}}\bigg(\cos\bigg(k|\mx-\mz|-\frac{n\pi}{2}-\frac{\pi}{4}\bigg)\bigg)=O\bigg(\frac{1}{\sqrt{k|\mx-\mz|}}\bigg),
\end{equation}
we can neglect the term
\[4\sum_{n=1}^{\infty}\frac{i^{n}}{n}J_n(kr)\sin\frac{1}{2}\bigg(n(\beta-\alpha)\bigg)\cos\frac{1}{2}\bigg(n(\beta+\alpha)\bigg).\]
This completes the proof.
\end{proof}

Based on Lemma \ref{Lemma}, we can immediately obtain the following results.
\begin{thm}[Single-frequency subspace migration]\label{Theorem1}
 For sufficiently large $N$, $\mathcal{I}(\mx;k)$ can be written as
  \begin{equation}\label{ImagingFunctionSingle}
\mathcal{I}(\mx;k)\approx\sum_{s=1}^{S}\left(J_0(k|\mx-\mz_s|)+O\bigg(\frac{1}{\sqrt{k|\mx-\mz|}}\bigg)\right)^2.
  \end{equation}
\end{thm}
\begin{proof}
Since the incident and observation direction configurations are the same, we set $\triangle\vt_p:=|\vt_p-\vt_{p-1}|$ for $p=2,3,\cdots,N,$ and $\triangle\vt_1:=|\vt_1-\vt_N|$. Then, applying Lemma \ref{Lemma} yields
\begin{align*}
\mathcal{I}(\mx;k)&=\left|\sum_{s=1}^{S}\left\langle\mathbf{W}(\mx;k),\mU_s(k)\right\rangle\left\langle\mathbf{W}(\mx;k),\overline{\mV}_s(k)\right\rangle\right|\simeq\sum_{s=1}^{S}\left(\sum_{p=1}^{N}e^{ik\vt_p\cdot(\mx-\mz_s)}\frac{\Delta\vt_p}{2\pi}\right)^2\\
&\approx\frac{1}{(\beta-\alpha)^{2}}\left(\int_{\mathbb{S}_{\mathrm{sub}}^1}e^{ik\vt\cdot(\mx-\mz_s)}dS(\vt)\right)^2=\frac{1}{(\beta-\alpha)^{2}}\sum_{s=1}^{S}\left((\beta-\alpha)J_0(k|\mx-\mz_s|)+O\bigg(\frac{1}{\sqrt{k|\mx-\mz|}}\bigg)\right)^{2}.
\end{align*}
Therefore, we can obtain
\[\mathcal{I}(\mx;k)\approx\sum_{s=1}^{S}\left(J_0(k|\mx-\mz_s|)+O\bigg(\frac{1}{\sqrt{k|\mx-\mz|}}\bigg)\right)^2.\]
This finishes the proof.
\end{proof}

Based on recent work (see \cite{A,AGKPS,HHSZ,JKHP,P1,P3,P4,PL2,PP}), it has been confirmed that a multi-frequency approach gives better results than applying a single frequency. The following theorem supports this fact.
\begin{thm}[Multi-frequency subspace migration]\label{Theorem2}
  Let $k_f:=2\pi/\lambda_f$ and $N$ be sufficiently large. Then, multi-frequency subspace migration
  \[\mathcal{I}_{\mathrm{MF}}(\mx,k_1,k_F;F)=\frac{1}{F}\left|\sum_{f=1}^{F}\sum_{s=1}^{S}\left\langle\mathbf{W}(\mx;k_f),\mU_s(k_f)\right\rangle\left\langle\mathbf{W}(\mx;k_f),\overline{\mV}_s(k_f)\right\rangle\right|\]
  can be written as follows
\begin{multline}\label{StructureImagingFunctionMultiple}
\mathcal{I}_{\mathrm{MF}}(\mx,k_1,k_F;F)\approx\left[\frac{k_F}{k_F-k_1}\bigg(J_0(k_F|\mx-\mz|)^{2}+J_1(k_F|\mx-\mz|)^{2}\bigg)\right.\\
\left.-\frac{k_1}{k_F-k_1}\bigg(J_0(k_1|\mx-\mz|)^{2}+J_1(k_1|\mx-\mz|)^{2}\bigg)\right].
  \end{multline}
\end{thm}
\begin{proof}
Similar to the proof of Lemma \ref{Lemma}, we consider polar coordinates such that $\vt=(\cos\theta,\sin\theta)$ and $\mx-\mz_s=r_s(\cos\phi_s,\sin\phi_s)$. Then, according to (\ref{ImagingFunctionSingle}), we can observe that
\[\mathcal{I}_{\mathrm{MF}}(\mx;k)\approx\frac{1}{F}\sum_{f=1}^{F}\sum_{s=1}^{S}J_0(k|\mx-\mz_s|)^{2}
\approx\frac{1}{k_F-k_1}\sum_{s=1}^{S}\int_{k_1}^{k_F}J_0(kr_s)^2dk.\]
Using this, we apply an indefinite integral formula of the Bessel function (see \cite[page 35]{R})
\[\int J_0(x)^2dx=x\bigg(J_0(x)^2+J_1(x)^2\bigg)+\int J_1(x)^2dx.\]
Thus, we immediately obtain
\begin{equation}\label{IntegralEquation}
\int_{k_1}^{k_F}J_0(kr_s)^2dk =k_F\bigg(J_0(k_Fr_s)^2+J_1(k_Fr_s)^2\bigg)-k_1\bigg(J_0(k_1r_s)^2+J_1(k_1r_s)^2\bigg)+\int_{k_1}^{k_F}J_1(kr_s)^2dk.
\end{equation}
Based on \cite[Theorem 3.4]{JKHP}, the last term of (\ref{IntegralEquation}) can be disregarded. Hence, we obtain (\ref{StructureImagingFunctionMultiple}).
\end{proof}

Based on Lemma 3.1, Theorems 3.2 and 3.3, we can examine certain properties. They can be summarized as follows.
\begin{enumerate}
  \item The condition of sufficiently large $k$ is a very strong assumption. If this condition is not fulfilled, the term
      \[O\bigg(\frac{1}{\sqrt{k|\mx-\mz|}}\bigg)\]
      of (\ref{ImagingFunctionSingle}) will disturb the imaging performance. This is why the application of a high frequency yields good results, and why subspace migration in the full-view problem offers better results than in the limited-view case.
  \item If $\mx$ is close to $\mz$, i.e., $|\mx-\mz|$ is small, then $k|\mx-\mz|$ is not sufficiently large and we cannot use the asymptotic form (\ref{AsymptoticBessel}). This means that the remaining term of Lemma \ref{Lemma}
      \[4\sum_{n=1}^{\infty}\frac{i^{n}}{n}J_n(kr_s)\sin\frac{1}{2}\bigg(n(\beta-\alpha)\bigg)\cos\frac{1}{2}\bigg(n(\beta+\alpha)\bigg)\]
      is not negligible, so some blurring effects will appear in the neighborhood of cracks (refer to Figure \ref{Result}).
  \item Conversely, if the location of $\mx$ is far from $\mz$, i.e., $|\mx-\mz|$ is large, then the map of $\mathcal{I}(\mx;k)$ will be plotted at $0$.
\end{enumerate}

\section{Numerical examples}\label{sec4}
In order to support some identified properties in Theorems \ref{Theorem1} and \ref{Theorem2}, we now present some numerical examples. For this purpose, $N=12$ different incident and observation directions are applied, such that
\[\vt_n=\alpha+(\beta-\alpha)\frac{n-1}{N-1}\quad\mbox{for}\quad\alpha=\frac{\pi}{4}\quad\mbox{and}\quad\beta=\frac{3\pi}{4}.\]
The applied wavenumber is of the form $k_f=2\pi/\lambda_f$, where $\lambda_f$, $f=1,2,\cdots,F(=10)$, is the given wavelength and $k_f$ are equi-distributed in the interval $[k_1,k_{10}]$ with $\lambda_1=0.6$ and $\lambda_{10}=0.2$.

The far-field pattern data $u_\infty^{(n)}(\mx;k)$ are generated by a Fredholm integral equation of the second kind along the crack introduced in \cite[Chapter 4]{N} to avoid committing inverse crimes. Three cracks $\Sigma_s$ with small length $\ell=0.05$ are chosen for the numerical simulations, with
\begin{align*}
  \Sigma_1&=\set{[t-0.6,-0.2]^T:-\ell\leq t\leq\ell}\\
  \Sigma_2&=\set{\mathfrak{R}_{\pi/4}[t+0.4,t+0.35]^T:-\ell\leq t\leq\ell}\\
  \Sigma_3&=\set{\mathfrak{R}_{7\pi/6}[t+0.25,t-0.6]^T:-\ell\leq t\leq\ell}.
\end{align*}
Here, $\mathfrak{R}_{\varphi}$ denotes rotation by $\varphi$.

\begin{figure}[!ht]
\begin{center}
\subfigure[Map of $\mathcal{I}(\mx,0.2)$]{\includegraphics[width=0.49\textwidth]{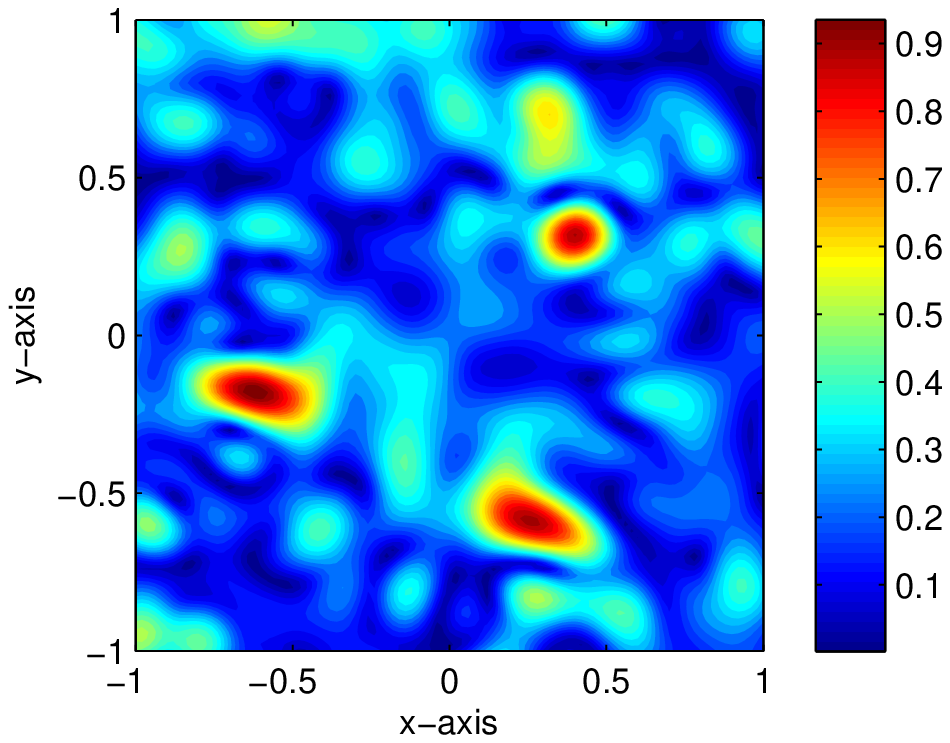}}
\subfigure[Map of $\mathcal{I}(\mx,0.2)$]{\includegraphics[width=0.49\textwidth]{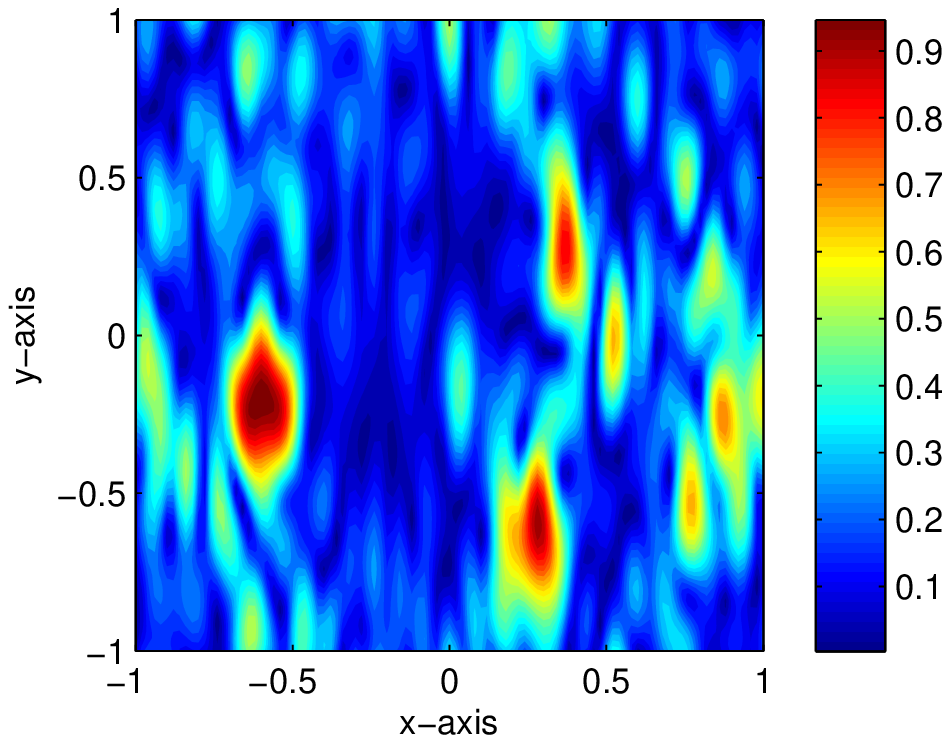}}\\
\subfigure[Map of $\mathcal{I}_{\mathrm{MF}}(\mx,k_1,k_{10};10)$]{\includegraphics[width=0.49\textwidth]{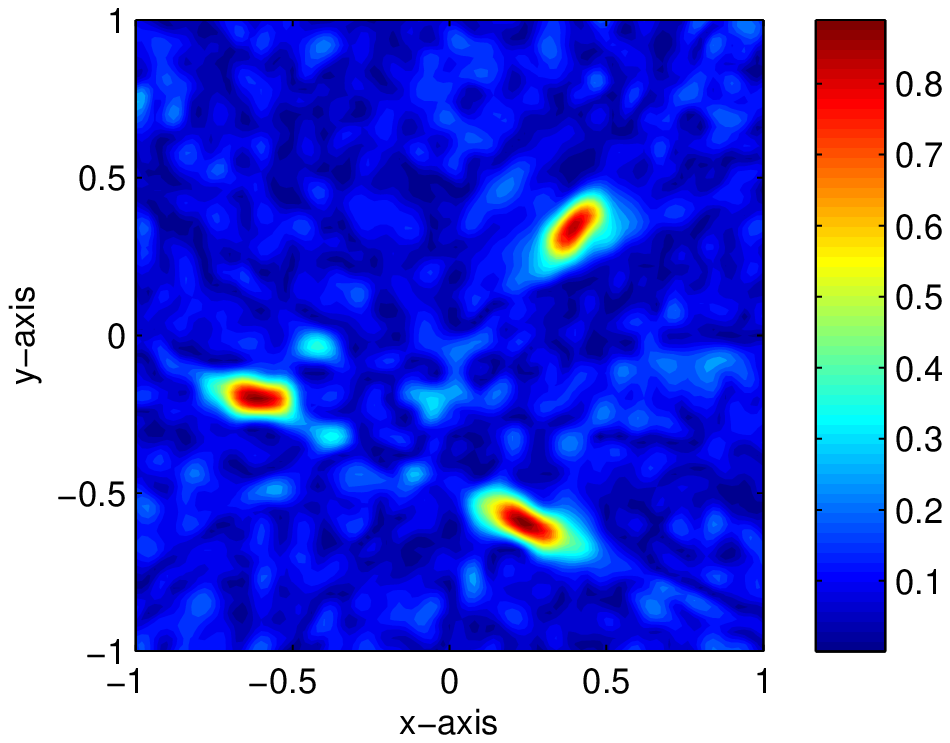}}
\subfigure[Map of $\mathcal{I}_{\mathrm{MF}}(\mx,k_1,k_{10};10)$]{\includegraphics[width=0.49\textwidth]{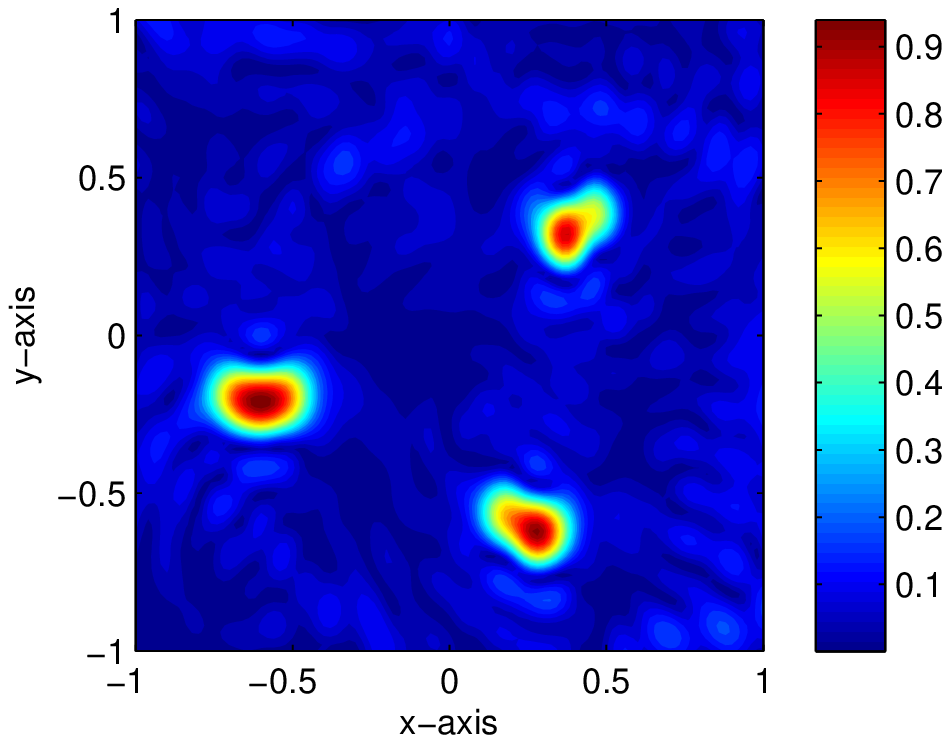}}
\caption{\label{Result}Maps of $\mathcal{I}(\mx,k_{10})$ and $\mathcal{I}_{\mathrm{MF}}(\mx,k_1,k_{10};10)$. Left column: full-view case. Right column: limited-view case.}
\end{center}
\end{figure}

Figure \ref{Result} shows maps of $\mathcal{I}(\mx,k_{10})$ and $\mathcal{I}_{\mathrm{MF}}(\mx,k_1,k_{10};10)$ for the full- and limited-view cases. Based on Theorem \ref{Theorem1}, the locations of the $\Sigma_s$ are successfully identified, but many replicas and blurring effects in the neighborhood of cracks disturb the identification. However, applying a multi-frequency yields a more accurate result than the single-frequency case. This supports Theorem \ref{Theorem2}. In any case, we can conclude that our analysis demonstrates that subspace migration can determine the location of small cracks in limited-view problems.

\section{Conclusions}
In this paper, we considered subspace migration for the imaging of small, perfectly conducting cracks. Based on a relation with an integer-order Bessel function of the first kind, we confirmed the effectiveness of subspace migration in limited-view inverse scattering problems under certain strong assumptions. We also identified some particular properties of the structure of this migration.

The main subject of this paper is the imaging of small cracks. Extension to the imaging of arbitrary shaped, arc-like cracks will be considered in future work. Although we considered a two-dimensional problem, its extension to  three dimensions will be an interesting challenge.

\section*{Acknowledgements}
This research was supported by the Basic Science Research Program through the National Research Foundation of Korea (NRF) funded by the Ministry of Education, Science and Technology (No. 2012-0003207), and the research program of Kookmin University in Korea.

\end{document}